\documentclass[11pt]{article}

\usepackage[T1]{fontenc}
\usepackage[utf8]{inputenc}
\usepackage{amsmath,amssymb,amsthm,mathtools,mathrsfs}
\usepackage{geometry}
\usepackage{hyperref}

\newtheorem{theorem}{Theorem}
\newtheorem{proposition}[theorem]{Proposition}
\newtheorem{corollary}[theorem]{Corollary}
\newtheorem{lemma}[theorem]{Lemma}

\theoremstyle{definition}

\theoremstyle{remark}
\newtheorem{remark}[theorem]{Remark}

\newcommand{\R}{\mathbb{R}}
\newcommand{\Tr}{\operatorname{Tr}}
\newcommand{\im}{\operatorname{im}}
\newcommand{\CS}{\operatorname{CS}}
\newcommand{\Var}{\operatorname{Var}}
\newcommand{\Cov}{\operatorname{Cov}}
\newcommand{\doi}[1]{\href{https://doi.org/#1}{doi:#1}}

\title{
Chern--Simons Fluctuations and Information Geometry\\
in Discrete Electromagnetism
}

\author{
Jean-Pierre Magnot
}

\date{}

\begin{document}

\maketitle

\begin{abstract}
We construct helicity-conditioned statistical states for a
Whitney-discretized electromagnetic field on a closed oriented
three-manifold. The simplicial de Rham complex provides exact
discrete gauge symmetry, while the Whitney inner product separates
exact, harmonic, and coexact sectors. The spatial Abelian
Chern--Simons functional is gauge invariant and depends only on the
coexact potential. After fixing harmonic modes, we introduce a
helicity-biased Gaussian ensemble on the reduced electromagnetic
phase space and derive explicit formulas for its admissible
parameters, partition function, mean helicity, relative entropy,
and Fisher information. The distribution uniquely minimizes
relative entropy under a prescribed mean-helicity constraint.
Its helicity susceptibility equals the variance of the discrete
Chern--Simons functional and controls the local distinguishability
of neighboring statistical states. Because magnetic helicity is
generally not conserved under unconstrained Maxwell dynamics, these
states represent conditioned inference rather than dynamical
equilibrium.
\end{abstract}

\medskip

\medskip

\section{Introduction}

Magnetic helicity plays a fundamental role in plasma physics,
magnetohydrodynamics, and the geometry of electromagnetic fields.
For a globally defined magnetic potential \(A\) on a closed
oriented three-manifold \(M\), it is given by
\begin{equation}
\mathscr H(A)
=
\int_M A\wedge dA.
\label{eq:continuous-helicity}
\end{equation}
In vector notation,
\begin{equation}
\mathscr H(A)
=
\int_M
A\cdot(\nabla\times A)\,d\operatorname{vol}.
\label{eq:vector-helicity}
\end{equation}

The relevance of helicity to force-free magnetic configurations
and relaxed plasma states goes back to Woltjer and Taylor
\cite{Woltjer1958,Taylor1974}. Its interpretation through
linkage and knottedness was developed by Moffatt
\cite{Moffatt1969}; see also \cite{ArnoldKhesin1998}.
The expression \eqref{eq:continuous-helicity} is simultaneously
the spatial Abelian instance of a Chern--Simons functional
\cite{ChernSimons1974}.

Helicity-constrained statistical constructions are not new.
In particular, Frisch, Pouquet, Léorat, and Mazure derived
absolute equilibrium spectra for spectrally truncated
magnetohydrodynamics while incorporating magnetic helicity
\cite{Frisch1975}. More generally, constrained statistical
inference naturally belongs to the maximum-entropy framework
introduced by Jaynes \cite{Jaynes1957,Jaynes1957b}.

The purpose of this paper is different. We formulate the
helicity constraint directly on a Whitney-discretized
electromagnetic complex, prove that it descends exactly through
discrete gauge reduction, and derive explicit mesh-level
expressions for its information-geometric fluctuations.

Whitney forms and finite element exterior calculus produce
discrete de Rham complexes preserving the cohomological
identities underlying gauge theories
\cite{ArnoldFalkWinther2006,ArnoldFalkWinther2010}.
These structures are particularly useful in computational
electromagnetism
\cite{Bossavit1998,Stern2015,Desbrun2005}.

For a simplicial one-cochain \(a\), we consider
\begin{equation}
\CS_K(a)
=
\int_M W_1(a)\wedge dW_1(a),
\label{eq:introduction-discrete-cs}
\end{equation}
where \(W_1\) is the Whitney map.

We show that \(\CS_K\) annihilates both exact gauge components
and harmonic components. After fixing the harmonic sector,
the remaining radiative phase space supports the probability
density
\begin{equation}
p_{\beta,\gamma}(a,e)
=
Z_K(\beta,\gamma)^{-1}
\exp\left(
-\beta E_K(a,e)
-\gamma\CS_K(a)
\right).
\label{eq:introduction-density}
\end{equation}
Here \(E_K\) denotes the reduced Maxwell energy, \(\beta>0\)
is the inverse-temperature parameter of a reference Maxwell
distribution, and \(\gamma\) is the multiplier conjugate to
the prescribed mean helicity.

The principal information-geometric identity is
\begin{equation}
\partial_\gamma^2\log Z_K
=
\Var_{\beta,\gamma}(\CS_K).
\label{eq:introduction-fisher}
\end{equation}
Consequently,
\begin{equation}
D\left(
p_{\beta,\gamma}
\,\middle\|\,
p_{\beta,\gamma+\varepsilon}
\right)
=
\frac{\varepsilon^2}{2}
\Var_{\beta,\gamma}(\CS_K)
+
O(\varepsilon^3).
\label{eq:introduction-relative}
\end{equation}

These fluctuation identities are familiar for exponential
families. Their specific contribution here is their exact
realization on a gauge-reduced Whitney electromagnetic
complex, together with computable mesh matrices and a
mode-resolved normalizability criterion.

A physical qualification is essential. Magnetic helicity alone
is not generally conserved by unconstrained Maxwell evolution.
Therefore, the distributions
\eqref{eq:introduction-density} are interpreted as
instantaneous helicity-conditioned inference states, not as
stationary equilibria of the free Maxwell dynamics.

Throughout, we work in the topologically trivial Abelian
bundle sector. Potentials are globally represented by
one-forms, and discrete gauge transformations are generated by
globally defined zero-cochains. No identification with
topologically massive Maxwell--Chern--Simons theory in
\(2+1\)-dimensional spacetime is intended.

\section{Whitney discretization and gauge reduction}

\subsection{Simplicial electromagnetic variables}

Let \(M\) be a closed oriented Riemannian three-manifold and
\(K\) a finite triangulation of \(M\). We write
\begin{equation}
C^k(K)
=
C^k(K;\R)
\end{equation}
for the space of real simplicial \(k\)-cochains.

The simplicial coboundary operators
\begin{equation}
\delta_k:
C^k(K)
\longrightarrow
C^{k+1}(K)
\end{equation}
form the complex
\begin{equation}
0
\longrightarrow
C^0(K)
\xrightarrow{\delta_0}
C^1(K)
\xrightarrow{\delta_1}
C^2(K)
\xrightarrow{\delta_2}
C^3(K)
\longrightarrow
0,
\label{eq:cochain-complex}
\end{equation}
with
\begin{equation}
\delta_{k+1}\delta_k=0.
\label{eq:delta-square}
\end{equation}

Let
\begin{equation}
W_k:
C^k(K)
\longrightarrow
\Omega^k_{\mathrm{pw}}(M)
\end{equation}
be the Whitney maps. They satisfy
\begin{equation}
dW_k
=
W_{k+1}\delta_k.
\label{eq:whitney-commutation}
\end{equation}

Whitney forms are understood as conforming finite element
differential forms. Integration-by-parts identities are used
in the weak sense; interior interface contributions cancel
because the relevant tangential traces agree across simplices.

A discrete magnetic potential is
\begin{equation}
a\in C^1(K),
\end{equation}
and its magnetic field is represented by
\begin{equation}
b=\delta_1a\in C^2(K).
\label{eq:magnetic-cochain}
\end{equation}

The discrete gauge transformation associated with
\(\chi\in C^0(K)\) is
\begin{equation}
a\longmapsto a+\delta_0\chi.
\label{eq:gauge-transformation}
\end{equation}
Equation \eqref{eq:delta-square} implies
\begin{equation}
\delta_1(a+\delta_0\chi)
=
\delta_1a.
\label{eq:magnetic-invariance}
\end{equation}

Thus the magnetic field is exactly invariant under the
discrete gauge symmetry.

\subsection{Whitney inner products and discrete Hodge sectors}

For each degree \(k\), define
\begin{equation}
\langle u,v\rangle_k
=
\int_M
W_k(u)\wedge\star W_k(v).
\label{eq:whitney-inner}
\end{equation}
The corresponding adjoint of \(\delta_k\) is denoted by
\begin{equation}
\delta_k^\ast:
C^{k+1}(K)
\longrightarrow
C^k(K).
\end{equation}

The degree-one harmonic space is
\begin{equation}
\mathcal H_K^1
=
\ker\delta_1
\cap
\ker\delta_0^\ast.
\label{eq:harmonic}
\end{equation}

Finite-dimensional Hodge theory gives
\begin{equation}
C^1(K)
=
\im\delta_0
\oplus
\mathcal H_K^1
\oplus
\im\delta_1^\ast,
\label{eq:hodge}
\end{equation}
with orthogonality for the Whitney inner product.

We introduce the notation
\begin{equation}
V_{\mathrm g}
=
\im\delta_0,
\qquad
V_{\mathrm h}
=
\mathcal H_K^1,
\qquad
V_{\mathrm T}
=
\im\delta_1^\ast.
\label{eq:sector-notation}
\end{equation}

The subscripts refer to gauge, harmonic, and transverse
radiative sectors, respectively.

Let \(e\in C^1(K)\) represent the electric variable using
the Riesz identification associated with
\eqref{eq:whitney-inner}. The discrete source-free Gauss law
is
\begin{equation}
\delta_0^\ast e=0.
\label{eq:gauss}
\end{equation}
Hence
\begin{equation}
e\in
V_{\mathrm h}\oplus V_{\mathrm T}.
\label{eq:e-sector}
\end{equation}

The reduced phase space is
\begin{equation}
\mathcal P_{\mathrm{red}}
=
\frac{
\left\{
(a,e)\in C^1(K)\oplus C^1(K):
\delta_0^\ast e=0
\right\}
}{
(a,e)\sim(a+\delta_0\chi,e)
}.
\label{eq:reduced-phase}
\end{equation}
Using \eqref{eq:hodge}, it identifies with
\begin{equation}
\mathcal P_{\mathrm{red}}
\simeq
\left(
V_{\mathrm h}\oplus V_{\mathrm T}
\right)
\oplus
\left(
V_{\mathrm h}\oplus V_{\mathrm T}
\right).
\label{eq:reduced-identification}
\end{equation}

The magnetic energy vanishes on harmonic potentials.
Therefore, integrating without additional constraints over
the harmonic potential sector would produce a nonnormalizable
Gaussian weight.

We fix both the harmonic potential sector and the harmonic
electric sector, and restrict the discussion to
\begin{equation}
\mathcal P_{\mathrm T}
=
V_{\mathrm T}\oplus V_{\mathrm T}.
\label{eq:radiative-phase}
\end{equation}

\begin{remark}
Harmonic modes are not gauge artifacts. They represent
cohomological degrees of freedom associated with
\(H^1(K;\R)\). Fixing them selects a radiative sector and
does not amount to quotienting them out as gauge directions.
The electric harmonic modes could alternatively be retained;
their Gaussian integration would merely contribute an
independent factor
\[
\left(
\frac{2\pi}{\beta}
\right)^{\dim\mathcal H_K^1/2}
\]
to the partition function.
\end{remark}

\section{The discrete Abelian Chern--Simons pairing}

Define
\begin{equation}
c_K(a,b)
=
\int_M
W_1(a)\wedge dW_1(b),
\qquad
a,b\in C^1(K),
\label{eq:pairing}
\end{equation}
and
\begin{equation}
\CS_K(a)
=
c_K(a,a).
\label{eq:cs-definition}
\end{equation}

\begin{proposition}
\label{prop:gauge-harmonic}
The bilinear form \(c_K\) is symmetric and satisfies
\begin{equation}
c_K(\delta_0\chi,a)=0
\label{eq:exact-annihilation}
\end{equation}
for all \(\chi\in C^0(K)\) and \(a\in C^1(K)\).

It also satisfies
\begin{equation}
c_K(h,a)=0
\label{eq:harmonic-annihilation}
\end{equation}
for every \(h\in\mathcal H_K^1\).

Consequently, if
\begin{equation}
a=a_{\mathrm g}+a_{\mathrm h}+a_{\mathrm T}
\end{equation}
is the decomposition \eqref{eq:hodge}, then
\begin{equation}
\CS_K(a)
=
\CS_K(a_{\mathrm T}).
\label{eq:only-transverse}
\end{equation}
\end{proposition}

\begin{proof}
Let
\begin{equation}
\alpha=W_1(a),
\qquad
\eta=W_1(b).
\end{equation}
Since \(M\) has no boundary,
\begin{equation}
0
=
\int_M d(\alpha\wedge\eta)
=
\int_M d\alpha\wedge\eta
-
\int_M\alpha\wedge d\eta.
\end{equation}
Because forms of degrees two and one commute under the wedge
product,
\begin{equation}
d\alpha\wedge\eta
=
\eta\wedge d\alpha.
\end{equation}
Therefore,
\begin{equation}
\int_M\alpha\wedge d\eta
=
\int_M\eta\wedge d\alpha,
\end{equation}
which proves symmetry.

From \eqref{eq:whitney-commutation},
\begin{equation}
W_1(\delta_0\chi)
=
dW_0(\chi).
\end{equation}
Hence
\begin{align}
c_K(\delta_0\chi,a)
&=
\int_M
dW_0(\chi)\wedge dW_1(a)
\nonumber\\
&=
\int_M
d\left(
W_0(\chi)dW_1(a)
\right)
\nonumber\\
&=
0.
\end{align}

If \(h\in\mathcal H_K^1\), then
\begin{equation}
\delta_1h=0,
\end{equation}
and therefore
\begin{equation}
dW_1(h)
=
W_2(\delta_1h)
=
0.
\end{equation}
Using symmetry,
\begin{equation}
c_K(h,a)
=
c_K(a,h)
=
\int_M
W_1(a)\wedge dW_1(h)
=
0.
\end{equation}

Expanding \(c_K(a,a)\) along \eqref{eq:hodge}, all terms
containing \(a_{\mathrm g}\) or \(a_{\mathrm h}\) vanish.
This gives \eqref{eq:only-transverse}.
\end{proof}

\begin{remark}
The vanishing of the harmonic contribution relies on working
with globally defined Abelian potentials on a closed
manifold. Nontrivial principal bundles, large gauge
transformations, and boundary contributions require additional
structure and are not included here.
\end{remark}

Let
\begin{equation}
n=\dim V_{\mathrm T}.
\end{equation}
Choose a Whitney-orthonormal basis of \(V_{\mathrm T}\). In
that basis, there exists a real symmetric matrix
\begin{equation}
C=C^\top\in\R^{n\times n}
\end{equation}
such that
\begin{equation}
\CS_K(a)
=
a^\top Ca.
\label{eq:matrix-cs}
\end{equation}

The reduced magnetic energy is
\begin{equation}
E_{\mathrm m}(a)
=
\frac12
\langle\delta_1a,\delta_1a\rangle_2.
\label{eq:magnetic-energy}
\end{equation}
Thus
\begin{equation}
E_{\mathrm m}(a)
=
\frac12a^\top La
\label{eq:matrix-energy}
\end{equation}
for a symmetric matrix
\begin{equation}
L=L^\top.
\end{equation}

\begin{lemma}
\label{lem:positive}
The matrix \(L\) is positive definite.
\end{lemma}

\begin{proof}
If
\begin{equation}
a^\top La=0,
\end{equation}
then
\begin{equation}
\delta_1a=0.
\end{equation}
Since \(a\in V_{\mathrm T}=\im\delta_1^\ast\),
\begin{equation}
a\in
\ker\delta_1\cap\im\delta_1^\ast.
\end{equation}
However,
\begin{equation}
\ker\delta_1
=
(\im\delta_1^\ast)^\perp.
\end{equation}
Therefore \(a=0\), proving positive definiteness.
\end{proof}

\section{Helicity-conditioned electromagnetic states}

\subsection{Normalizability and partition function}

On
\begin{equation}
\mathcal P_{\mathrm T}
\simeq
\R^n\oplus\R^n,
\end{equation}
define the radiative Maxwell energy
\begin{equation}
E_K(a,e)
=
\frac12e^\top e
+
\frac12a^\top La.
\label{eq:total-energy}
\end{equation}

For
\begin{equation}
\beta>0,
\qquad
\gamma\in\R,
\end{equation}
consider
\begin{equation}
p_{\beta,\gamma}(a,e)
=
Z_K(\beta,\gamma)^{-1}
\exp\left(
-\beta E_K(a,e)
-
\gamma\CS_K(a)
\right),
\label{eq:density}
\end{equation}
whenever the normalization factor exists.

Introduce
\begin{equation}
Q_{\beta,\gamma}
=
\beta L+2\gamma C.
\label{eq:Q}
\end{equation}

\begin{theorem}
\label{thm:partition}
The partition function
\begin{equation}
Z_K(\beta,\gamma)
=
\int_{\R^n}
\int_{\R^n}
\exp\left(
-\beta E_K(a,e)
-
\gamma\CS_K(a)
\right)
\,da\,de
\label{eq:Z-definition}
\end{equation}
is finite if and only if
\begin{equation}
Q_{\beta,\gamma}>0.
\label{eq:Q-positive}
\end{equation}

On the admissible set
\begin{equation}
\mathscr D_K
=
\left\{
(\beta,\gamma)\in(0,\infty)\times\R:
\beta L+2\gamma C>0
\right\},
\label{eq:admissible}
\end{equation}
one has
\begin{equation}
Z_K(\beta,\gamma)
=
(2\pi)^n
\beta^{-n/2}
\det(Q_{\beta,\gamma})^{-1/2}.
\label{eq:Z-explicit}
\end{equation}
\end{theorem}

\begin{proof}
Using \eqref{eq:total-energy} and \eqref{eq:matrix-cs},
\begin{equation}
\beta E_K(a,e)+\gamma\CS_K(a)
=
\frac{\beta}{2}e^\top e
+
\frac12a^\top Q_{\beta,\gamma}a.
\end{equation}
Consequently,
\begin{align}
Z_K(\beta,\gamma)
&=
\left(
\int_{\R^n}
e^{-\beta e^\top e/2}\,de
\right)
\nonumber\\
&\quad\times
\left(
\int_{\R^n}
e^{-a^\top Q_{\beta,\gamma}a/2}\,da
\right).
\end{align}

The first factor is finite because \(\beta>0\).
The second factor is finite precisely when
\(Q_{\beta,\gamma}\) is positive definite.

The standard Gaussian formulas give
\begin{equation}
\int_{\R^n}
e^{-\beta e^\top e/2}\,de
=
(2\pi)^{n/2}\beta^{-n/2}
\end{equation}
and
\begin{equation}
\int_{\R^n}
e^{-a^\top Q_{\beta,\gamma}a/2}\,da
=
(2\pi)^{n/2}
\det(Q_{\beta,\gamma})^{-1/2}.
\end{equation}
Their product is \eqref{eq:Z-explicit}.
\end{proof}

\begin{corollary}
\label{cor:spectral}
Let
\begin{equation}
T=L^{-1/2}CL^{-1/2}
\label{eq:T}
\end{equation}
and let
\begin{equation}
\lambda_1,\dots,\lambda_n
\end{equation}
be the eigenvalues of \(T\).

Then
\begin{equation}
(\beta,\gamma)\in\mathscr D_K
\end{equation}
if and only if
\begin{equation}
\beta+2\gamma\lambda_j>0
\label{eq:eigenvalue-condition}
\end{equation}
for every \(j=1,\dots,n\).

Furthermore,
\begin{equation}
Z_K(\beta,\gamma)
=
(2\pi)^n
\beta^{-n/2}
\det(L)^{-1/2}
\prod_{j=1}^n
(\beta+2\gamma\lambda_j)^{-1/2}.
\label{eq:Z-eigenvalues}
\end{equation}
\end{corollary}

\begin{proof}
We have
\begin{equation}
Q_{\beta,\gamma}
=
L^{1/2}
\left(
\beta I+2\gamma T
\right)
L^{1/2}.
\end{equation}
Since \(L^{1/2}\) is invertible, \(Q_{\beta,\gamma}>0\) if
and only if
\begin{equation}
\beta I+2\gamma T>0.
\end{equation}
Diagonalizing \(T\) yields
\eqref{eq:eigenvalue-condition}.

Taking determinants,
\begin{equation}
\det(Q_{\beta,\gamma})
=
\det(L)
\prod_{j=1}^n
(\beta+2\gamma\lambda_j).
\end{equation}
Substitution into \eqref{eq:Z-explicit} proves
\eqref{eq:Z-eigenvalues}.
\end{proof}

For fixed \(\beta>0\), the admissible helicity multipliers
form the open interval
\begin{equation}
I_\beta
=
\bigcap_{j=1}^n
\left\{
\gamma\in\R:
\beta+2\gamma\lambda_j>0
\right\}.
\label{eq:I-beta}
\end{equation}
This interval always contains \(\gamma=0\).

If both positive and negative generalized helicity
eigenvalues occur, \(I_\beta\) is bounded on both sides.
If all nonzero eigenvalues have the same sign, it is a
half-line.

\subsection{Mean energy, mean helicity, and entropy}

Expectation with respect to \(p_{\beta,\gamma}\) is denoted
by
\begin{equation}
\langle\cdot\rangle_{\beta,\gamma}.
\end{equation}

\begin{proposition}
\label{prop:means}
On \(\mathscr D_K\),
\begin{equation}
\langle E_K\rangle_{\beta,\gamma}
=
-\partial_\beta\log Z_K
=
\frac{n}{2\beta}
+
\frac12
\Tr\left(
Q_{\beta,\gamma}^{-1}L
\right)
\label{eq:mean-energy}
\end{equation}
and
\begin{equation}
\langle\CS_K\rangle_{\beta,\gamma}
=
-\partial_\gamma\log Z_K
=
\Tr\left(
Q_{\beta,\gamma}^{-1}C
\right).
\label{eq:mean-CS}
\end{equation}

Equivalently,
\begin{equation}
\langle\CS_K\rangle_{\beta,\gamma}
=
\sum_{j=1}^n
\frac{\lambda_j}{
\beta+2\gamma\lambda_j
}.
\label{eq:mean-CS-eigenvalues}
\end{equation}
\end{proposition}

\begin{proof}
Equation \eqref{eq:Z-explicit} gives
\begin{equation}
\log Z_K
=
n\log(2\pi)
-
\frac{n}{2}\log\beta
-
\frac12\log\det(Q_{\beta,\gamma}).
\label{eq:log-Z}
\end{equation}
For an invertible matrix \(Q\),
\begin{equation}
\partial_\theta\log\det Q
=
\Tr\left(
Q^{-1}\partial_\theta Q
\right).
\end{equation}
Since
\begin{equation}
\partial_\beta Q_{\beta,\gamma}=L
\end{equation}
and
\begin{equation}
\partial_\gamma Q_{\beta,\gamma}=2C,
\end{equation}
differentiating \eqref{eq:log-Z} yields
\eqref{eq:mean-energy} and \eqref{eq:mean-CS}.

Formula \eqref{eq:mean-CS-eigenvalues} follows either from
\eqref{eq:Z-eigenvalues} or by diagonalizing \(T\).
\end{proof}

The differential entropy is
\begin{equation}
S_K(\beta,\gamma)
=
-\int_{\mathcal P_{\mathrm T}}
p_{\beta,\gamma}
\log p_{\beta,\gamma}
\,da\,de.
\label{eq:entropy-definition}
\end{equation}
From \eqref{eq:density},
\begin{equation}
S_K(\beta,\gamma)
=
\log Z_K
+
\beta\langle E_K\rangle_{\beta,\gamma}
+
\gamma\langle\CS_K\rangle_{\beta,\gamma}.
\label{eq:entropy-identity}
\end{equation}

Moreover,
\begin{align}
\beta\langle E_K\rangle_{\beta,\gamma}
+
\gamma\langle\CS_K\rangle_{\beta,\gamma}
&=
\frac{n}{2}
+
\frac12
\Tr\left(
Q_{\beta,\gamma}^{-1}
(\beta L+2\gamma C)
\right)
\nonumber\\
&=
n.
\end{align}
Therefore,
\begin{equation}
S_K(\beta,\gamma)
=
n\log(2\pi e)
-
\frac{n}{2}\log\beta
-
\frac12\log\det(Q_{\beta,\gamma}).
\label{eq:entropy-explicit}
\end{equation}

\begin{remark}
The differential entropy depends on the Lebesgue measure
associated with a chosen Whitney-orthonormal basis. Relative
entropy between two probability measures is independent of
such coordinate normalization and therefore provides the more
intrinsic statistical quantity.
\end{remark}

\section{Minimum relative entropy under a helicity constraint}

For probability densities \(p,q\) on
\(\mathcal P_{\mathrm T}\), write
\begin{equation}
D(p\|q)
=
\int_{\mathcal P_{\mathrm T}}
p\log\frac{p}{q}
\,da\,de.
\label{eq:KL}
\end{equation}

Fix \(\beta>0\), and introduce the unconstrained Maxwell
reference distribution
\begin{equation}
p_{\beta,0}
=
Z_K(\beta,0)^{-1}
e^{-\beta E_K}.
\label{eq:reference}
\end{equation}

\begin{theorem}
\label{thm:min-relative}
Let
\begin{equation}
\gamma\in I_\beta
\end{equation}
and define
\begin{equation}
h_\gamma
=
\langle\CS_K\rangle_{\beta,\gamma}.
\label{eq:h-gamma}
\end{equation}

Among all probability densities \(p\) satisfying
\begin{equation}
\int_{\mathcal P_{\mathrm T}}
p\,\CS_K\,da\,de
=
h_\gamma
\label{eq:constraint}
\end{equation}
and
\begin{equation}
D(p\|p_{\beta,0})<\infty,
\end{equation}
the density \(p_{\beta,\gamma}\) uniquely minimizes
\begin{equation}
D(p\|p_{\beta,0}).
\end{equation}

More precisely,
\begin{equation}
D(p\|p_{\beta,0})
=
D(p\|p_{\beta,\gamma})
+
D(p_{\beta,\gamma}\|p_{\beta,0})
\label{eq:pythagoras}
\end{equation}
for every admissible \(p\) satisfying
\eqref{eq:constraint}.
\end{theorem}

\begin{proof}
From \eqref{eq:density} and \eqref{eq:reference},
\begin{equation}
\log
\frac{
p_{\beta,\gamma}
}{
p_{\beta,0}
}
=
-\gamma\CS_K
-
\log Z_K(\beta,\gamma)
+
\log Z_K(\beta,0).
\label{eq:ratio}
\end{equation}
Thus
\begin{align}
D(p\|p_{\beta,0})
-
D(p\|p_{\beta,\gamma})
&=
\int
p
\log
\frac{
p_{\beta,\gamma}
}{
p_{\beta,0}
}
\nonumber\\
&=
-\gamma h_\gamma
-
\log Z_K(\beta,\gamma)
+
\log Z_K(\beta,0).
\label{eq:difference}
\end{align}

Applying the same calculation with
\begin{equation}
p=p_{\beta,\gamma}
\end{equation}
gives
\begin{equation}
D(p_{\beta,\gamma}\|p_{\beta,0})
=
-\gamma h_\gamma
-
\log Z_K(\beta,\gamma)
+
\log Z_K(\beta,0).
\end{equation}
Combining this equality with \eqref{eq:difference} yields
\eqref{eq:pythagoras}.

Since
\begin{equation}
D(p\|p_{\beta,\gamma})\geq0
\end{equation}
with equality only when
\begin{equation}
p=p_{\beta,\gamma}
\end{equation}
almost everywhere, uniqueness follows.
\end{proof}

\begin{corollary}
\label{cor:strict}
If \(C\neq0\), then
\begin{equation}
\gamma
\longmapsto
\langle\CS_K\rangle_{\beta,\gamma}
\end{equation}
is strictly decreasing on \(I_\beta\).

Consequently, every mean helicity in the image of this map
determines a unique helicity-conditioned inference state.
\end{corollary}

\begin{proof}
Differentiating \eqref{eq:mean-CS-eigenvalues} gives
\begin{equation}
\frac{d}{d\gamma}
\langle\CS_K\rangle_{\beta,\gamma}
=
-2
\sum_{j=1}^n
\frac{
\lambda_j^2
}{
(\beta+2\gamma\lambda_j)^2
}.
\label{eq:strict-monotonicity}
\end{equation}
If \(C\neq0\), at least one eigenvalue \(\lambda_j\) is
nonzero. Hence the derivative is strictly negative.
\end{proof}

For two admissible parameter pairs,
\begin{align}
&D\left(
p_{\beta_1,\gamma_1}
\,\middle\|\,
p_{\beta_2,\gamma_2}
\right)
\nonumber\\
&\quad=
\log Z_K(\beta_2,\gamma_2)
-
\log Z_K(\beta_1,\gamma_1)
\nonumber\\
&\qquad+
(\beta_2-\beta_1)
\langle E_K\rangle_{\beta_1,\gamma_1}
\nonumber\\
&\qquad+
(\gamma_2-\gamma_1)
\langle\CS_K\rangle_{\beta_1,\gamma_1}.
\label{eq:full-relative}
\end{align}

In particular, at fixed \(\beta\),
\begin{align}
&D\left(
p_{\beta,\gamma_1}
\,\middle\|\,
p_{\beta,\gamma_2}
\right)
\nonumber\\
&\quad=
\log Z_K(\beta,\gamma_2)
-
\log Z_K(\beta,\gamma_1)
\nonumber\\
&\qquad+
(\gamma_2-\gamma_1)
\langle\CS_K\rangle_{\beta,\gamma_1}.
\label{eq:fixed-relative}
\end{align}

This relative entropy does not measure helicity itself.
Instead, it quantifies the distinguishability of two
statistical descriptions obtained by imposing different mean
helicity constraints on the same reduced electromagnetic
system.

\section{Information geometry and helicity susceptibility}

Let
\begin{equation}
\theta^1=\beta,
\qquad
\theta^2=\gamma.
\end{equation}
The Fisher information matrix is
\begin{equation}
g_{ij}
=
\int_{\mathcal P_{\mathrm T}}
p_{\beta,\gamma}
\partial_i\log p_{\beta,\gamma}
\partial_j\log p_{\beta,\gamma}
\,da\,de.
\label{eq:fisher}
\end{equation}

\begin{proposition}
\label{prop:fisher}
On \(\mathscr D_K\),
\begin{equation}
g_{ij}
=
\partial_i\partial_j
\log Z_K(\beta,\gamma).
\label{eq:hessian}
\end{equation}

The entries are
\begin{equation}
g_{\beta\beta}
=
\Var_{\beta,\gamma}(E_K)
=
\frac{n}{2\beta^2}
+
\frac12
\Tr\left(
Q_{\beta,\gamma}^{-1}L
Q_{\beta,\gamma}^{-1}L
\right),
\label{eq:gbb}
\end{equation}
\begin{equation}
g_{\beta\gamma}
=
\Cov_{\beta,\gamma}(E_K,\CS_K)
=
\Tr\left(
Q_{\beta,\gamma}^{-1}L
Q_{\beta,\gamma}^{-1}C
\right),
\label{eq:gbg}
\end{equation}
and
\begin{equation}
g_{\gamma\gamma}
=
\Var_{\beta,\gamma}(\CS_K)
=
2\Tr\left(
Q_{\beta,\gamma}^{-1}C
Q_{\beta,\gamma}^{-1}C
\right).
\label{eq:ggg}
\end{equation}

Equivalently,
\begin{equation}
g_{\gamma\gamma}
=
2
\sum_{j=1}^n
\frac{
\lambda_j^2
}{
(\beta+2\gamma\lambda_j)^2
}.
\label{eq:g-eigenvalues}
\end{equation}
\end{proposition}

\begin{proof}
Differentiating \eqref{eq:density} yields
\begin{equation}
\partial_\beta\log p_{\beta,\gamma}
=
-E_K
+
\langle E_K\rangle_{\beta,\gamma}
\end{equation}
and
\begin{equation}
\partial_\gamma\log p_{\beta,\gamma}
=
-\CS_K
+
\langle\CS_K\rangle_{\beta,\gamma}.
\end{equation}
Substitution into \eqref{eq:fisher} identifies the Fisher
matrix with the covariance matrix of
\begin{equation}
(E_K,\CS_K).
\end{equation}

Differentiation under the Gaussian integral gives
\eqref{eq:hessian}. To obtain the explicit expressions, use
\begin{equation}
\partial_\theta Q^{-1}
=
-Q^{-1}
(\partial_\theta Q)
Q^{-1}
\end{equation}
together with
\begin{equation}
\partial_\beta Q=L,
\qquad
\partial_\gamma Q=2C.
\end{equation}

Finally, \eqref{eq:g-eigenvalues} follows from
\eqref{eq:strict-monotonicity}.
\end{proof}

\begin{corollary}
\label{cor:local}
For fixed \(\beta>0\) and \(\gamma\in I_\beta\),
\begin{equation}
D\left(
p_{\beta,\gamma}
\,\middle\|\,
p_{\beta,\gamma+\varepsilon}
\right)
=
\frac{\varepsilon^2}{2}
\Var_{\beta,\gamma}(\CS_K)
+
O(\varepsilon^3)
\label{eq:local}
\end{equation}
as \(\varepsilon\to0\).

If \(C\neq0\), the quadratic coefficient is strictly
positive.
\end{corollary}

\begin{proof}
Expand
\begin{equation}
\log Z_K(\beta,\gamma+\varepsilon)
\end{equation}
at \(\gamma\):
\begin{align}
\log Z_K(\beta,\gamma+\varepsilon)
&=
\log Z_K(\beta,\gamma)
\nonumber\\
&\quad+
\varepsilon
\partial_\gamma\log Z_K(\beta,\gamma)
\nonumber\\
&\quad+
\frac{\varepsilon^2}{2}
\partial_\gamma^2\log Z_K(\beta,\gamma)
+
O(\varepsilon^3).
\end{align}
Substitution into \eqref{eq:fixed-relative} cancels the
linear term because
\begin{equation}
\partial_\gamma\log Z_K
=
-\langle\CS_K\rangle_{\beta,\gamma}.
\end{equation}
The remaining quadratic coefficient is
\begin{equation}
\frac12g_{\gamma\gamma}
=
\frac12
\Var_{\beta,\gamma}(\CS_K).
\end{equation}
\end{proof}

Define the helicity susceptibility by
\begin{equation}
\chi_{\mathrm H}
=
-\partial_\gamma
\langle\CS_K\rangle_{\beta,\gamma}.
\label{eq:susceptibility-definition}
\end{equation}
Then
\begin{equation}
\boxed{
\chi_{\mathrm H}
=
\Var_{\beta,\gamma}(\CS_K)
=
g_{\gamma\gamma}.
}
\label{eq:main-identity}
\end{equation}

Thus the same quantity governs the response to changes in
the helicity multiplier, Chern--Simons fluctuations, and
local statistical distinguishability.

\section{Whitney matrices and paired helicity modes}

\subsection{Mesh-level expressions}

Let
\begin{equation}
\{w_i\}_{i=1}^{N_1}
\end{equation}
be the Whitney edge basis.

Denote by \(M_1\) and \(M_2\) the degree-one and degree-two
Whitney mass matrices. Let \(D_0\) and \(D_1\) represent the
incidence operators \(\delta_0\) and \(\delta_1\).

Define
\begin{equation}
S_{ij}
=
\int_M
w_i\wedge dw_j.
\label{eq:S}
\end{equation}
Proposition \ref{prop:gauge-harmonic} gives
\begin{equation}
S=S^\top
\label{eq:S-symmetric}
\end{equation}
and
\begin{equation}
SD_0=0.
\label{eq:S-gauge}
\end{equation}

Choose a matrix
\begin{equation}
B\in\R^{N_1\times n}
\end{equation}
whose columns form a Whitney-orthonormal basis of
\(V_{\mathrm T}\):
\begin{equation}
B^\top M_1B=I_n.
\label{eq:B}
\end{equation}
Then
\begin{equation}
L
=
B^\top D_1^\top M_2D_1B
\label{eq:L}
\end{equation}
and
\begin{equation}
C
=
B^\top SB.
\label{eq:C}
\end{equation}

The admissibility condition becomes
\begin{equation}
\beta B^\top D_1^\top M_2D_1B
+
2\gamma B^\top SB
>
0.
\label{eq:mesh-admissibility}
\end{equation}

Hence every quantity derived above is determined by
simplicial incidence data, Whitney mass matrices, and the
Chern--Simons pairing \eqref{eq:S}.

\subsection{An idealized paired-helicity block}

Introduce the magnetic-energy normalized variable
\begin{equation}
x=L^{1/2}a.
\label{eq:x}
\end{equation}
Then
\begin{equation}
E_{\mathrm m}(a)
=
\frac12x^\top x
\label{eq:normalized-energy}
\end{equation}
and
\begin{equation}
\CS_K(a)
=
x^\top Tx,
\qquad
T=L^{-1/2}CL^{-1/2}.
\label{eq:normalized-CS}
\end{equation}

Suppose that \(T\) possesses two eigenvalues
\begin{equation}
\lambda_+=\lambda,
\qquad
\lambda_-=-\lambda,
\qquad
\lambda>0.
\label{eq:paired}
\end{equation}

On the corresponding two-mode subspace,
\begin{equation}
\CS_K(x_+,x_-)
=
\lambda(x_+^2-x_-^2).
\label{eq:two-helicity}
\end{equation}

This is an idealized paired-helicity block. Such spectral
pairing is not asserted for arbitrary Whitney triangulations.

The potential-sector Gaussian weight is
\begin{equation}
\exp\left[
-\frac12
(\beta+2\gamma\lambda)x_+^2
-
\frac12
(\beta-2\gamma\lambda)x_-^2
\right].
\end{equation}

Normalizability requires
\begin{equation}
|\gamma|
<
\frac{\beta}{2\lambda}.
\label{eq:paired-domain}
\end{equation}

The two-mode contribution to the mean helicity is
\begin{align}
h_{\mathrm{pair}}(\beta,\gamma)
&=
\frac{\lambda}{
\beta+2\gamma\lambda
}
-
\frac{\lambda}{
\beta-2\gamma\lambda
}
\nonumber\\
&=
-\frac{
4\gamma\lambda^2
}{
\beta^2-4\gamma^2\lambda^2
}.
\label{eq:paired-helicity}
\end{align}

The corresponding susceptibility is
\begin{equation}
\chi_{\mathrm{pair}}(\beta,\gamma)
=
2\lambda^2
\left[
\frac{1}{
(\beta+2\gamma\lambda)^2
}
+
\frac{1}{
(\beta-2\gamma\lambda)^2
}
\right].
\label{eq:paired-susceptibility}
\end{equation}

At \(\gamma=0\),
\begin{equation}
h_{\mathrm{pair}}(\beta,0)=0
\end{equation}
but
\begin{equation}
\chi_{\mathrm{pair}}(\beta,0)
=
\frac{4\lambda^2}{\beta^2}.
\label{eq:paired-zero}
\end{equation}

As
\begin{equation}
|\gamma|
\uparrow
\frac{\beta}{2\lambda},
\end{equation}
one helicity-polarized Gaussian mode loses confinement and
the susceptibility diverges.

At fixed finite mesh, this divergence is not a thermodynamic
phase transition. It identifies the boundary of
normalizability of the conditioned Gaussian ensemble.

\section{Compatibility with Maxwell dynamics}

The reduced Maxwell Hamiltonian is
\begin{equation}
E_K(a,e)
=
\frac12e^\top e
+
\frac12a^\top La.
\end{equation}

With the standard canonical convention,
\begin{equation}
\dot a=e,
\qquad
\dot e=-La.
\label{eq:maxwell-flow}
\end{equation}

\begin{proposition}
\label{prop:dynamics}
Along the reduced Maxwell evolution
\eqref{eq:maxwell-flow},
\begin{equation}
\frac{d}{dt}\CS_K(a(t))
=
2a(t)^\top Ce(t).
\label{eq:CS-evolution}
\end{equation}

In general, \(\CS_K\) is not conserved. Moreover,
\begin{equation}
\frac{d^2}{dt^2}\CS_K(a(t))
=
2e(t)^\top Ce(t)
-
a(t)^\top(CL+LC)a(t).
\label{eq:CS-second}
\end{equation}
\end{proposition}

\begin{proof}
Since \(C=C^\top\),
\begin{equation}
\frac{d}{dt}
\left(
a^\top Ca
\right)
=
\dot a^\top Ca
+
a^\top C\dot a
=
2a^\top C\dot a.
\end{equation}
Using \(\dot a=e\) gives
\eqref{eq:CS-evolution}.

Differentiating again,
\begin{align}
\frac{d^2}{dt^2}
\left(
a^\top Ca
\right)
&=
2e^\top Ce
+
2a^\top C\dot e
\nonumber\\
&=
2e^\top Ce
-
2a^\top CLa.
\end{align}
Because
\begin{equation}
a^\top CLa
=
\frac12
a^\top(CL+LC)a,
\end{equation}
formula \eqref{eq:CS-second} follows.
\end{proof}

The continuous analogue is
\begin{equation}
\frac{d}{dt}
\int_M A\wedge dA
=
2\int_M
\partial_tA\wedge dA.
\label{eq:continuous-evolution}
\end{equation}
With the physical convention
\begin{equation}
E=-\partial_tA
\end{equation}
in temporal gauge, this becomes
\begin{equation}
\frac{d}{dt}
\int_M A\cdot B\,d\operatorname{vol}
=
-2\int_M E\cdot B\,d\operatorname{vol}.
\label{eq:EB}
\end{equation}

The sign difference between
\eqref{eq:CS-evolution} and \eqref{eq:EB} simply reflects the
choice of canonical momentum convention.

Since
\begin{equation}
a^\top Ce
\end{equation}
does not vanish identically, the distributions
\(p_{\beta,\gamma}\) with \(\gamma\neq0\) are not generally
stationary under the unconstrained Maxwell flow.

They must therefore be understood as minimum-relative-entropy
states conditioned on instantaneous helicity information.
A dynamical equilibrium interpretation would require an
additional helicity-preserving hypothesis, as occurs in
appropriate ideal magnetohydrodynamic settings.

\section{Discussion and conclusion}

The Whitney discretization separates three distinct
structures.

First, gauge redundancy is encoded by the exact cochain
sector
\begin{equation}
\im\delta_0.
\end{equation}
Because
\begin{equation}
\delta_1\delta_0=0,
\end{equation}
the magnetic field and the Abelian Chern--Simons functional
are invariant under the corresponding discrete gauge
transformations.

Second, harmonic cochains represent genuine cohomological
degrees of freedom. They are not gauge directions, but they
do not contribute to the Abelian Chern--Simons functional in
the globally trivial closed-manifold setting considered
here. Fixing the harmonic sector yields a strictly positive
radiative Maxwell energy.

Third, the Chern--Simons pairing defines a symmetric
quadratic observable on the coexact sector. Conditioning the
Maxwell reference distribution on its mean produces the
unique minimum-relative-entropy state
\begin{equation}
p_{\beta,\gamma}
=
Z_K(\beta,\gamma)^{-1}
e^{-\beta E_K-\gamma\CS_K}.
\end{equation}

Its existence is characterized exactly by
\begin{equation}
\beta L+2\gamma C>0.
\end{equation}

The corresponding partition function, mean helicity,
relative entropy, and Fisher information are computable
directly from Whitney mass matrices, incidence matrices, and
the discrete Chern--Simons pairing.

The principal fluctuation relation is
\begin{equation}
\boxed{
-\partial_\gamma
\langle\CS_K\rangle_{\beta,\gamma}
=
\Var_{\beta,\gamma}(\CS_K)
=
g_{\gamma\gamma}.
}
\label{eq:final}
\end{equation}
Thus helicity response, Chern--Simons fluctuations, and local
statistical distinguishability coincide.

Although the Chern--Simons pairing contains no explicit
Hodge star, the resulting statistical model is not a
metric-independent topological field theory. The Maxwell
energy, Whitney mass matrices, radiative decomposition, and
partition function depend on the Riemannian metric and the
chosen triangulation.

The construction also does not claim mesh-refinement
invariance, a continuum limit, or stationarity under free
Maxwell dynamics. These issues, as well as non-Abelian
extensions, quantization, and boundary contributions, require
additional analysis.
\vskip 12pt

\paragraph{\bf Data availability statement} No data is available for this work.

\vskip 12pt

\paragraph{\bf Conflict of interest statement} The author declares no conflict of interest.

\vskip 12pt

\paragraph{\bf Funding} No funding supported this work.

\vskip 12pt

\paragraph{\bf Acknowledgements} J.-P.M thanks the France 2030 framework programme Centre Henri Lebesgue ANR-11-LABX-0020-01 
for creating an attractive mathematical environment.

\vskip 12pt

\paragraph{\bf Author's Note on AI Assistance}
Portions of the text were developed with the assistance of a generative language model (OpenAI ChatGPT, based on the GPT-4 architecture). The AI was used to assist with drafting, editing, and standardizing the bibliography format. All mathematical content, structure, and theoretical constructions were provided, verified, and curated by the author. The author assumes full responsibility for the correctness, originality, and scholarly integrity of the final manuscript.


\begin{thebibliography}{99}

\bibitem{Woltjer1958}
L. Woltjer,
A theorem on force-free magnetic fields,
\emph{Proc. Natl. Acad. Sci. USA}
\textbf{44} (1958), 489--491.
\doi{10.1073/pnas.44.6.489}

\bibitem{Taylor1974}
J. B. Taylor,
Relaxation of toroidal plasma and generation of reverse
magnetic fields,
\emph{Phys. Rev. Lett.}
\textbf{33} (1974), 1139--1141.
\doi{10.1103/PhysRevLett.33.1139}

\bibitem{Moffatt1969}
H. K. Moffatt,
The degree of knottedness of tangled vortex lines,
\emph{J. Fluid Mech.}
\textbf{35} (1969), 117--129.
\doi{10.1017/S0022112069000991}

\bibitem{ArnoldKhesin1998}
V. I. Arnold and B. A. Khesin,
\emph{Topological Methods in Hydrodynamics},
Applied Mathematical Sciences, Vol. 125,
Springer, New York, 1998.
\doi{10.1007/b97593}

\bibitem{ChernSimons1974}
S.-S. Chern and J. Simons,
Characteristic forms and geometric invariants,
\emph{Ann. of Math.}
\textbf{99} (1974), 48--69.
\doi{10.2307/1971013}

\bibitem{Frisch1975}
U. Frisch, A. Pouquet, J. Léorat, and A. Mazure,
Possibility of an inverse cascade of magnetic helicity in
magnetohydrodynamic turbulence,
\emph{J. Fluid Mech.}
\textbf{68} (1975), 769--778.
\doi{10.1017/S002211207500122X}

\bibitem{Jaynes1957}
E. T. Jaynes,
Information theory and statistical mechanics,
\emph{Phys. Rev.}
\textbf{106} (1957), 620--630.
\doi{10.1103/PhysRev.106.620}

\bibitem{Jaynes1957b}
E. T. Jaynes,
Information theory and statistical mechanics. II,
\emph{Phys. Rev.}
\textbf{108} (1957), 171--190.
\doi{10.1103/PhysRev.108.171}

\bibitem{ArnoldFalkWinther2006}
D. N. Arnold, R. S. Falk, and R. Winther,
Finite element exterior calculus, homological techniques, and
applications,
\emph{Acta Numer.}
\textbf{15} (2006), 1--155.
\doi{10.1017/S0962492906210018}

\bibitem{ArnoldFalkWinther2010}
D. N. Arnold, R. S. Falk, and R. Winther,
Finite element exterior calculus: from Hodge theory to
numerical stability,
\emph{Bull. Amer. Math. Soc.}
\textbf{47} (2010), 281--354.
\doi{10.1090/S0273-0979-10-01278-4}

\bibitem{Bossavit1998}
A. Bossavit,
\emph{Computational Electromagnetism:
Variational Formulations, Complementarity, Edge Elements},
Academic Press, San Diego, 1998.

\bibitem{Stern2015}
A. Stern, Y. Tong, M. Desbrun, and J. E. Marsden,
Geometric computational electrodynamics with variational
integrators and discrete differential forms,
in
\emph{Geometry, Mechanics, and Dynamics:
The Legacy of Jerry Marsden},
Fields Institute Communications, Vol. 73,
Springer, New York, 2015, pp. 437--475.
\doi{10.1007/978-1-4939-2441-7\_19}

\bibitem{Desbrun2005}
M. Desbrun, A. N. Hirani, M. Leok, and J. E. Marsden,
Discrete exterior calculus,
arXiv:math/0508341, 2005.

\end{thebibliography}
\end{document}